\documentclass[a4paper]{article}

\usepackage{amsmath}
\usepackage[inline]{enumitem}
\usepackage{multirow}
\usepackage{amssymb}
\usepackage{amsthm}
\usepackage{graphicx}
\usepackage{bbm}
\usepackage{color}
\usepackage{authblk}

\newcommand{\bra}[1]{\left\langle#1\right|}
\newcommand{\ket}[1]{\left|#1\right\rangle}
\newcommand{\braket}[2]{\left\langle#1\middle|#2\right\rangle}

\newcommand{\unit}{1\kern -3.7pt 1}

\DeclareMathOperator{\tr}{tr}
\DeclareMathOperator{\sgn}{sgn}

\theoremstyle{definition}
\newtheorem{defn}{Definition}[section]

\theoremstyle{plain}
\newtheorem{lem}[defn]{Lemma}
\newtheorem{thm}[defn]{Theorem}
\newtheorem{prop}[defn]{Proposition}

\theoremstyle{remark}
\newtheorem*{remark}{Remark}

\begin{document}

\title{$n$-fold unbiased bases: an extension of the MUB condition}
\author{M\'at\'e Farkas}

\affil{\textit{Institute of Theoretical Physics and Astrophysics, National Quantum Information Centre, Faculty of Mathematics, Physics and Informatics, University of Gda\'nsk, 80-952 Gdansk, Poland}}

\maketitle

\begin{abstract}
I introduce a new notion, that extends the mutually unbiased bases
(MUB) conditons to more than two bases. 
These, I call the nUB conditions, and the corresponding bases $n$-fold
unbiased. They naturally appear while optimizing generic $n$-to-one quantum
random access code (QRAC) strategies.
While their existence in general dimensions is an open question, they
nevertheless give close-to-tight upper bounds on QRAC success probabilities, and
raise fundamental questions about the geometry of quantum states.
\end{abstract}

\section{Introduction}

Mutually unbiased bases (MUBs) are an important notion in quantum information
theory, first studied in the context of optimal
state-determination \cite{ivonovic,wooters_fields}. Later,
they found applications in
entropic uncertainty relations (\cite{maassen_uffink}, and surveys
\cite{uncertainty_survey,uncertainty_survey2}), information locking
\cite{locking1,locking2} and the so-called Mean King's problem
\cite{mean_king1,mean_king2}.
Intuitively speaking, if some classical information is encoded in a basis, then
measuring in a basis unbiased to it reveals nothing about the encoded
information whatsoever (see Section \ref{sec:QRAC} for a formal definition).
There is a great number of papers investigating the existence and constructions
of these bases (see \cite{onMUBs} for a survey, \cite{MUB2-5} for a
classification in dimensions
2-5 and \cite{MUB6_3,MUB6_4,MUB6_1,MUB6_2,MUB6_5,MUB6_6} for the question
of the number of MUBs in
dimension 6). It is known that in any
dimension, there are at least 3, and at most $d+1$ MUBs, the upper bound
being saturated in prime power dimensions. Composite dimensions on the other
hand still remain unsolved.

While the bases are called mutually unbiased, the MUB conditions on $n$ bases
effectively impose only \textit{pairwise} mutual unbiasedness. In this paper, I
introduce a new
notion, which is a global constraint on $n$ bases, that I call \textit{
$n$-fold unbiased bases} ($n$UBs). These conditions naturally arise, while
extending some methods of \cite{QRACMUB}. There, the authors prove
that MUBs provide optimal measurements in
the so-called quantum random access code (QRAC) protocol in the two-input case.
$n$-fold unbiased bases then generalize this optimization task to $n$ inputs.

The above mentioned QRACs are a basic information theoretical
protocol, used in many contexts within quantum information theory (for a
comprehensive generic description, see \cite{QRACSR}). Loosely speaking, the
task is to compress $n$ dits into one (quantum) dit, and to be able to recover
one randomly choosen dit with high probability (see Section \ref{sec:QRAC} for
a formal treatment). First, it appeared in \cite{conjugate_coding}, and was
called conjugate coding. Later, it was studied in the context of quantum finite
automata \cite{finite_automata2,finite_automata1,finite_automata3}, quantum
communication complexity
\cite{complexity1,complexity2,complexity3,complexity4}, network coding
\cite{network1,network2}, and locally decodable
codes \cite{locally1,locally2,locally3,locally4}. Recently, it is used also for
``quantumness
witness'', that is, experimentally distinguishing different product structures
of fixed dimensional systems \cite{QRAC1024}. Its versatile use is the
consequence of its simplicity, and the fact that it provides quantum advantage
over classical strategies.

The paper is organized as follows: in Section \ref{sec:QRAC}, I formally
describe $n^d\to1$ QRACs, and cite the result stating that MUBs are optimal in a
$2^d\to1$ QRAC scenario. In Section \ref{sec:nUB}, I give a formal definition
of $n$-fold unbiased bases, and state my main theorem about their optimality in
$n^d\to1$ QRACs, whenever $d\ge n$. Then, in Section \ref{sec:methods}, I give
a rigorous proof of the main theorem. In Section \ref{sec:existence}, I address
the problem of existence of $n$-fold unbiased bases with some rigorous results
in low dimensions, but leaving the general question open. Section
\ref{sec:applications} focuses on applications, mainly considering QRACs, but
also outlining some other potential applications. Finally, in Section
\ref{sec:foundations} I consider two foundational issues connected to the $n$UB
construction: the geometry of quantum states, and the question of genuine $n$-th
order interference.

\section{Quantum random access codes}\label{sec:QRAC}

The short description of an $n^d\to1$ quantum random access code (QRAC) is as
follows (see Fig. \ref{fig:qrac}). Alice is given a
classical input $x=\{x_1,x_2,\ldots,x_n\}$, which is a string of dits, i.e.
$x_i\in[d]$, where I use the notation $[d]=\{1,2,\ldots,d\}$. Alice
then is allowed to send one $d$-dimensional (quantum) state to Bob, denoted by
$\rho_x$, depending on her input. Bob is given a classical input $y\in[n]$, and
his task is to guess $x_y$. Generally, he makes his guess by
performing a measurement $\mathcal{M}^y$ on the state, depending on his input,
where $\mathcal{M}^y=\{ M_b^y\}_{b=1}^{d}$.
The measurement satisfies the usual conditions:
$\sum_{b=1}^dM^y_b=\unit$ and $M_b^y\ge0$. The usual question is: what
states and measurements give the optimal strategy for a QRAC? By optimality, in
the following, I mean maximal average success probability (ASP):

\begin{equation}\label{eq:asp}
\bar{p}=\frac1{nd^n}\sum_{x,y}\mathbb{P}(B=x_y\vert X=x,Y=y)=\frac1{nd^n}\sum_
{x,y}\tr(\rho_xM^y_{x_y}),
\end{equation}
where the capital letters denote the probabilistic variables of the
corresponding lower-case symbols, and $x$ and $y$ run along all their possible
values (I implicitly assume unifrom distribution on the inputs, see 
\cite{QRACSR}).

\begin{figure}[h!]
\begin{center}
\includegraphics[height=3cm]{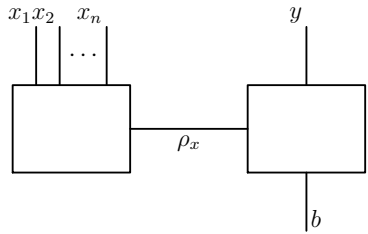}
\caption{Schematic representation of the $n^d\to1$ QRAC protocol.\label{fig:qrac}}
\end{center}
\end{figure}

Some cases are already well-studied \cite{QRACMUB}, and it is proven that
in the $n=2$ case, mutually unbiased measurements (or, more abstractly, mutually
unbiased bases, MUBs) give the optimal strategy. This can always
be done, since there exists a pair of MUBs in any dimension. For the readers'
convenience, I recall the definition of MUBs.

\begin{defn}Consider two orthogonal bases on $\mathbb{C}^d$, $\{\ket{y_i}\}_
{i=1}^d$ and $\{\ket{z_j}\}_{j=1}^d$. We say that these bases are mutually
unbiased, if they satisfy
\begin{equation}\label{eq:MUB}
\lvert \braket{y_i}{z_j} \rvert=\frac{1}{\sqrt{d}} ~~~ \forall i,j\in[d].
\end{equation}
\end{defn}

Using this definition, the following theorem is proven in \cite{QRACMUB}:

\begin{thm}\label{thm:main}For a $2^d\to1$ QRAC, the optimal
strategy is obtained by measuring in MUBs of dimension $d$.
\end{thm}

\section{$n$-fold unbiased bases}\label{sec:nUB}

In the following, using a similar line of argument with which it's proven that
MUBs are optimal for $2^d\to1$ QRACs, I will show that a natural generalization
of these bases provide optimal strategies for $n^d\to1$ QRACs. Let me give the
definition of the mentioned generalization. Later on, I will show that this
condition arises naturally in the QRAC scenario.

\begin{defn}\label{def:nUB}Consider $n$ orthogonal bases on $\mathbb{C}^d$,
$\{\ket{y_{x_y}}\}_
{x_y=1}^d$, where $y=1,\ldots,n$. We say that these bases are \textit{$n$-fold
unbiased}, if they satisfy
\begin{equation}\label{eq:nUB}
\sum_{\substack{\sigma\in S_n \\ \sigma\text{: n-cycle}}}\prod_{y=1}^n
\braket{y_{x_y}}{\sigma(y)_
{x_{\sigma(y)}}}=\frac{(n-1)!}{d^{n-1}} ~~~ \forall x_1,\ldots,x_n\in[d].
\end{equation}
\end{defn}

\begin{remark}
Here $\sigma$ is an element of the permutation group $S_n$. Note that the
essence of this criterion is that these terms should be uniform for each
$x_1,\ldots,x_n\in[d]$. The particular value comes from the restriction when
we sum up over all $x_1,\ldots,x_n$, and that there are $(n-1)!$ $n$-cycles in
$S_n$. Also note, that for $n=2$, we get the MUB condition, Eq.
(\ref{eq:MUB}).
\end{remark}

Now, my main result concerning QRAC strategies is the following:

\begin{thm}\label{thm:main}For an $n^d\to1$ QRAC with $d\ge n$, the optimal
strategy is obtained by measuring in $n\text{UB}$s of dimension $d$.
\end{thm}

In the next section, I provide the methods for proving the above theorem.

\section{Methods}\label{sec:methods}

In order to prove the main theorem, several results are needed on QRAC
strategies.
The following lemmas allow us to only use pure states on both the encoding and
the decoding sides:

\begin{lem}\label{lem:pure_encoding}
For an $n^d\to1$ QRAC, pure state encoding is sufficient to reach an optimal
strategy.
\end{lem}
\begin{proof}
See \cite{QRACSR}.
\end{proof}

This means, that in fact $\rho_x=\ket{\psi_x}\bra{\psi_x}$, a pure state on
$\mathbb{C}^d$.
Next, it is shown in \cite{QRAC1024}, that von Neumann measurements are
optimal.

\begin{lem}\label{lem:pure_measurement}
For an $n^d\to1$ QRAC, von Neumann measurements are sufficient to reach an
optimal strategy.
\end{lem}
\begin{proof}
See \cite{QRAC1024}.
\end{proof}

Which means, that in fact, $\{M^y_b\}_b=\{\ket{y_b}\bra{y_b}\}_b$,
where $\ket{y_b}\in\mathbb{C}^d$ and $\sum_b\ket{y_b}\bra{y_b}
=\unit$ for each $y$.

It is then rather straightforward, and shown in \cite{QRACMUB}, that for any
set of measurements on Bob's side, the optimal encoding for an input $x$ is the
eigenvector $\ket{\psi_x}$ of the operator
\begin{equation}\label{eq:Mx_op}
M_x=\sum_y\ket{y_{x_y}}\bra{y_{x_y}},
\end{equation}
that corresponds to the largest eigenvalue $\lambda^{\text
{max}}_x$. The ASP then becomes

\begin{equation}\label{eq:asp_eigenvalue}
\bar{p}=\frac1{nd^n}\sum_x\lambda^{\text{max}}_x.
\end{equation}

The task is to maximize this expression by choosing optimal measurements. I
will concentrate on the characteristic polynomial of $M_x$, since its zeroes
give (among other eigenvalues) $\lambda_x^{\text{max}}$. First, note that if
$d\ge n$, in the optimal case we can assume that the vectors $\{\ket{y_
{x_y}}\}_y$ span an $n$-dimensional subspace in $\mathbb{C}^d$ for every $x$.
This is because otherwise the optimization for every $x$ is restricted to a
lower dimensional subspace, giving in general suboptimal results. We can then
consider this set of vectors a (not necessarily orthogonal) basis for this
subspace, and write the matrix of $M_x$ in this basis. It is easy to see that
this will be the Gramian matrix of the set $\{\ket{y_{x_y}}\}_y$, i.e. in
this basis, $(M_x)_{yy'}=\langle y_{x_y}\vert{y'}_{x_{y'}}\rangle$.
For now, I will suppress the index $x$ for notational simplicity, and analyse
the eigenvalue $\lambda^{\text{max}}$ of the operator $M=\sum_y\ket{y
}\bra{y}$.

The characteristic polynomial in general takes the form
\begin{equation}\label{eq:charpoly_general}
\mathcal{P}(\lambda)=\lambda^n+c_1\lambda^{n-1}+c_2\lambda^{n-2}+\cdots+c_n,
\end{equation}
where the coefficients $c_k$ can be written as
\begin{equation}\label{eq:ck}
c_k=(-1)^k\sum_{\vert J\vert=k}M[J],
\end{equation}
where $M[J]$ is the principal minor of the matrix $M$, that corresponds to the
set $J\subseteq\{1,\ldots,n\}$. So, for example, $c_1=-\tr{M}$ and $c_n=
(-1)^n\det{M}$.

The following lemma is crucial in obtaining the optimal measurement bases:
\begin{lem}\label{lem:concavity}
The maximal eigenvalue $\lambda^{\text{max}}$ of the operator $M=\sum_y\ket
{y}\bra{y}$ is a concave function of all the coefficients $c_k$ in the
characteristic polynomial, expressed by Eq. (\ref{eq:ck}).
\end{lem}
\begin{proof}
First, analyse the characteristic polynomial, now only as a function of
$c_k$, assuming all other coefficients to be constant. (Note that by varying
$c_k$, in reality, we're altering all other $c_l$ coefficients, as we are
altering the measurement bases. For now, I forget about this fact, and am
looking for purely the best solution based on a generic characteristic
polynomial): 
\begin{equation}\label{eq:charpoly_ck}
\mathcal{P}_{c_k}(\lambda)=\lambda^n+c_1\lambda^{n-1}+c_2\lambda^{n-2}+\cdots+
c_n.
\end{equation}
It depends linearly on all coefficients $c_k$, thus
any series expansion is of first order:
\begin{equation}\label{eq:charpoly_expansion}
\mathcal{P}_{c_k+\text{d}c_k}(\lambda)=\mathcal{P}_{c_k}(\lambda)+\lambda^{n-k}
\text{d}c_k,
\end{equation}
where $\text{d}c_k$ is an infinitesimal change in $c_k$. Let's now call
$\tilde{
\lambda}^{\text{max}}$ the maximal zero of this modified polynomial, i.e.
$\mathcal{P}_{c_k+\text{d}c_k}(\tilde{\lambda}^{\text{max}})=0$, whereas
$\mathcal{P}_{c_k}(\lambda^{\text{max}})=0$ from the original problem. I am
interesed in the concavity of $\lambda^{\text{max}}$ in $c_k$, i.e. the sign of
the second derivative $\frac{\partial^2\lambda^{\text{max}}}{\partial c_k^2}$.
For this, expand $\tilde{\lambda}^{\text{max}}$ up to second order:
\begin{equation}\label{eq:lambda_expansion}
\tilde{\lambda}^{\text{max}}=\lambda^{\text{max}}+\text{d}\lambda^{\text{max}}=
\lambda^{\text{max}}+\frac{\partial\lambda^
{\text{max}}}{\partial c_k}\text{d}c_k+\frac12\frac{\partial^2\lambda^{\text
{max}}} {\partial c_k^2}\text{d}c_k^2+\mathcal{O}(\text{d}c_k^3).
\end{equation}
Also, expand $\mathcal{P}_{c_k}(\lambda)$ in $\lambda$ to second order,
around $\lambda^{\text{max}}$:
\begin{equation}\label{eq:charpoly_lambda}
\begin{split}
\mathcal{P}_{c_k}(\tilde{\lambda}^{\text{max}})& \left. =\mathcal{P}_{c_k}
(\lambda^ {\text{max}})+\frac{\partial\mathcal{P}}{\partial\lambda}\bigg\rvert_
{c_k,\lambda^ {\text{max}}} \text{d}\lambda^{\text{max}}+\frac{\partial^2
\mathcal{P}}{\partial\lambda^2}\bigg\rvert_{c_k,\lambda^ {\text{max}}}
(\text{d}\lambda^{\text{max}})^2 \right. \\
& \left. = \frac{\partial\mathcal{P}}{\partial\lambda}\bigg\rvert_
{c_k,\lambda^ {\text{max}}} \text{d}\lambda^{\text{max}}+\frac{\partial^2
\mathcal{P}}{\partial\lambda^2}\bigg\rvert_{c_k,\lambda^ {\text{max}}}
(\text{d}\lambda^{\text{max}})^2, \right.
\end{split}
\end{equation}
as the first term vanishes. Now, evaluate Eq. (\ref{eq:charpoly_expansion}) at
$\tilde{\lambda}^{\text{max}}$ up to second order, using Eqs.
(\ref{eq:lambda_expansion}) and (\ref{eq:charpoly_lambda}):
\begin{equation}\label{eq:second_order_eq}
\begin{split}
0 & \left. =\frac{\partial\mathcal{P}}{\partial\lambda}\bigg\rvert_{c_k,\lambda^
{\text{max}}}\bigg(\frac{\partial\lambda^{\text{max}}}{\partial c_k}\text{d}c_k
+\frac12
\frac{\partial^2\lambda^{\text{max}}}{\partial c_k^2}\text{d}c_k^2\bigg)+\frac12
\frac
{\partial^2\mathcal{P}}{\partial\lambda^2}\bigg\rvert_{c_k,\lambda^{\text{max}}}
\bigg(\frac{\partial\lambda^{\text{max}}}{\partial c_k}\bigg)^2\text{d}c_k^2
\right. \\
& \left. + (\tilde{\lambda}^{\text{max}})^{n-k}\text{d}c_k. \right.
\end{split}
\end{equation}
Since d$c_k$ is an arbitrary infinitesimal, the terms multiplying d$c_k$ and
d$c_k^2$ should be equal independently, yielding the following two equations,
respectively:
\begin{equation}\label{eq:dc_k}
\frac{\partial\mathcal{P}}{\partial\lambda}\bigg\rvert_{c_k,\lambda^
{\text{max}}}\frac{\partial\lambda^{\text{max}}}{\partial c_k}=-(\tilde{\lambda}
^{\text {max}})^{n-k}
\end{equation}
\begin{equation}\label{eq:dc_k^2}
\frac{\partial\mathcal{P}}{\partial\lambda}\bigg\rvert_{c_k,\lambda^
{\text{max}}}\frac{\partial^2\lambda^{\text{max}}}{\partial c_k^2}=-\frac
{\partial^2\mathcal{P}}{\partial\lambda^2}\bigg\rvert_{c_k,\lambda^{\text{max}}}
\bigg(\frac{\partial\lambda^{\text{max}}}{\partial c_k}\bigg)^2.
\end{equation}
Then, analyse the derivatives of $\mathcal{P}(\lambda)$ at $\lambda^
{\text{max}}$. Remember, that $\mathcal{P}(\lambda)$ is the characteristic
polynomial of the operator $M$ in Eq. (\ref{eq:Mx_op}), which is positive. This
means that all its eigenvalues (i.e. the zeroes of $\mathcal{P}(\lambda)$) are
real positive numbers. Then, invoke the Gauss--Lucas theorem (see e.g.
\cite[Theorem 6.1]{gauss_lucas}), that says that
for a polynomial $\mathcal{P}(\lambda)$, all the zeroes of $\frac{\partial
\mathcal{P}}{\partial\lambda}$ belong to the
convex hull of the set of zeroes of $\mathcal{P}$. For us, this means that none
of the derivatives $\frac{\partial\mathcal{P}}{\partial\lambda}$ and
$\frac{\partial^2\mathcal{P}}{\partial\lambda^2}$ change signs outside of the
region $[\lambda^{\text{min}},\lambda^{\text{max}}]$. It is clear from Eq.
(\ref{eq:charpoly_ck}) that both of these derivatives are positive in the limit
$\lambda\to\infty$. For the moment, let's assume that $\lambda^{\text{max}}$ is
a nondegenerate zero of $\mathcal{P}$. Then it follows that both $\frac{\partial
\mathcal{P}}{\partial\lambda}\big\rvert_{c_k,\lambda^{\text{max}}}$ and
$\frac{\partial^2\mathcal{P}}{\partial\lambda^2}\big\rvert_{c_k,\lambda^
{\text{max}}}$ are strictly positive. Using this in Eq. (\ref{eq:dc_k^2}), it
follows that $\frac{\partial^2\lambda^{\text{max}}}{\partial c_k^2}<0$,
proving the lemma.

In the case where $\lambda^{\text{max}}$ is a degenerate zero of $\mathcal{P}$,
we need to investigate also Eq. (\ref{eq:dc_k}). The RHS is always negative,
while $\frac{\partial\mathcal{P}}{\partial\lambda}\big\rvert_{c_k,\lambda^
{\text{max}}}$ is positive in the non-degenerate case. In the degenerate case it
is 0, and the right and left derivatives could differ in the sense that they
could be equal to $0^+$ or $0^-$ (limiting from above or below). This means that
the derivative $\frac{\partial\lambda^{\text{max}}}{\partial c_k}$ equals
$\pm\infty$. Since we are dealing with a polynomial and its zeroes, every
function appearing is smooth. Thus, when approaching the degenerate case by
varying $c_k$, the derivative of $\lambda^{\text{max}}$ cannot suddenly
change from a negative value to $+\infty$. Hence, we can conclude that it equals
$-\infty$, and the derivative $\frac{\partial\mathcal{P}}{\partial\lambda}\big
\rvert_{c_k,\lambda^{\text{max}}}$ that should be considered is the right
derivative, which is always positive. The same observations hold for Eq.
(\ref{eq:dc_k^2}), and we can conclude that $\frac{\partial^2\lambda^
{\text{max}}}{\partial c_k^2}$ is negative in the degenerate case as well.
\end{proof}

Now, I re-introduce the index $x$, and write the maximal eigenvalue for a
given $x$ as a function of the coefficients in the corresponding characteristic
polynomial: $\lambda_x^{\text{max}}\big((c_1)_x,(c_2)_x,\ldots,(c_n)_x\big)$.
Then, we can write the ASP (Eq.
(\ref{eq:asp_eigenvalue})), as a function of the vectors $\boldsymbol{c}_k$,
that contain $(c_k)_x$ for all $x$:
\begin{equation}\label{eq:asp_vector}
\bar{p}(\boldsymbol{c}_1,\boldsymbol{c}_2,\ldots,\boldsymbol{c}_n)=\frac1{nd^n}
\sum_x\lambda^{\text{max}}_x\big((c_1)_x,(c_2)_x,\ldots,(c_n)_x\big).
\end{equation}

The following lemma is very useful to characterize the behaviour of
the ASP as a function of the vectors $\boldsymbol{c}_k$:
\begin{lem}\label{lem:ck_constant}
Consider the operators $M_x$ described by Eq. (\ref{eq:Mx_op}), and the
corresponding coefficients $(c_k)_x$ in the characteristic polynomial,
as in
Eq. (\ref{eq:ck}). Then for every $k$, the sum of $(c_k)_x$ over all $x$ is a
constant dependent only on the dimension.
\end{lem}
\begin{proof}
From Eq. (\ref{eq:ck}) it is enough to show the statement for principal minors
of the $M_x$ operators. I.e. let me fix $J\subseteq\{1,\ldots,n\}$, $\vert
J\vert =k$, and show that $\sum_{x}M_x[J]$ is a constant, only depending on $d$.
$M_x[J]$ is the Gram determinant of the vectors $\{\lvert y_{x_y}\rangle\}
_{y\in J}$, thus it can be written as
\begin{equation}\label{eq:M_x[J]}
M_x[J]=\sum_{\sigma\in S_k}\sgn(\sigma)\prod_{y\in J}
\braket{y_{x_y}}{{\sigma(y)}_{x_{\sigma(y)}}},
\end{equation}
where $\sigma$ runs over all permutations on $J$. Every permutation $\sigma$
can be decomposed to disjoint cycles $\sigma=\sigma_1\sigma_2\cdots\sigma_r$ on
the disjoint sets $J_1,J_2,\ldots,J_r\subseteq J$ and $\cup_{i=1}^rJ_i=J$. Then
the above expression becomes:
\begin{equation}\label{eq:M_x[J]_decomp}
M_x[J]=\sum_{\sigma\in S_k}\sgn(\sigma)\prod_{y\in J_1}
\braket{y_{x_y}}{{\sigma_1(y)}_{x_{\sigma_1(y)}}}\prod_{y\in J_2}
\braket{y_{x_y}}{{\sigma_2(y)}_{x_{\sigma_2(y)}}}\cdots\prod_{y\in
J_r} \braket{y_{x_y}}{{\sigma_r(y)}_{x_{\sigma_r(y)}}}.
\end{equation}

When we sum up over $x$, it is equivalent to saying that we sum up over all
$x_y$. Hence in the above expression, we can evaluate the sum over
$x$ independently on each product over $y\in J_i$, summing up over all $x_y$
such that $y\in J_i$. Then, for proving the lemma, I have to show that all
expressions of the form
\begin{equation}\label{eq:constterm}
m_i=\sum_{x_y:y\in J_i}\prod_{y\in J_i}
\braket{y_{x_y}}{{\sigma_i(y)}_{x_{\sigma_i(y)}}}
\end{equation}
are constant, depending only on $d$. Let's say that $\sigma_i$ is a $k_i$-cycle.
Clearly, if $k_i=1$, then $m_i=d$. Now,
for $k_i>1$, we can always order the product in a way that the ket and bra
of the same $y_{x_y}$ appear next to each other. Explicitly, it means that
with this ordering:
\begin{equation}\label{eq:constterm2}
m_i=\sum_{x_y:y\in J_i}
\langle{y^\ast}_{x_{y^\ast}}\lvert{\sigma_i(y^\ast)}_{x_{\sigma_i(y
^\ast)}}\rangle\langle{\sigma_i(y^\ast)}_{x_{\sigma_i(y^\ast)}}\lvert
{\sigma^2_i(y^\ast)}_{x_{\sigma^2_i(y^\ast)}}\rangle\langle
{\sigma^2_i(y^\ast)}_{x_{\sigma^2_i(y^\ast)}}\rvert\cdots\lvert
{y^\ast}_{x_{y^\ast}}\rangle,
\end{equation}
for some arbitrary $y^\ast\in J_i$. Now, when summing up over any $x_y$ such
that $y\ne y^\ast$, we get that $\sum_{x_y}\lvert{y}_{x_{y}}\rangle\langle
{y}_{x_{y}}\rvert=\unit$ from the orthonormality of the measurement bases.
Thus
\begin{equation}\label{eq:constterm3}
m_i=\sum_{x_{y^\ast}=1}^d\langle{y^\ast}_{x_{y^\ast}}\lvert{y^\ast}_
{x_{y^\ast}}\rangle=d,
\end{equation}
and the proof is complete.
\end{proof}

It is also important to note that since $M_x$ is positive, all its prinicpal
minors are positive. This means, that the sign of $(c_k)_x$ is $(-1)^k$ for all
$x$. This, together with Lemma \ref{lem:ck_constant} means that we can look at
the vectors $\boldsymbol{c}_k$ as probability distributions normalized to
some constants that depend only on the dimension. Thus, the ASP in Eq. (\ref{eq:asp_vector}) is a function on probability distributions,
and from Lemma \ref{lem:concavity} it follows that it is a Schur-concave function of all
the vectors $\boldsymbol{c}_1,\boldsymbol{c}_2,\ldots,\boldsymbol{c}_n$
\cite[Theorems A.3, A.4.]{inequalities}. But then it follows that it is maximized by all such
distributions set uniform \cite[Proposition B.2.]{inequalities}.

Now, let's discuss what does the uniformity of these vectors mean.
If we follow the argument in the proof of Lemma \ref{lem:ck_constant}, we see
that the only terms that appear in $\boldsymbol{c}_k$, but are not present in
any $\boldsymbol{c}_l$ with $l<k$ are
\begin{equation}\label{eq:ck_newterm}
\sum_{\substack{\sigma\in S_k \\ \sigma\text{: k-cycle}}}\prod_{y\in J}
\braket{y_{x_y}}{{\sigma(y)}_{x_{\sigma(y)}}}=:k\text{UB}_x[J],
\end{equation}
for all $J\subseteq\{1,\ldots,n\}$, $\vert J\vert=k$ (this is because if
$\sigma$ is not a $k$-cycle, it can be decomposed to disjoint $l<k$ cycles, and
the terms we obtain from this decomposition have already appeared in 
$\boldsymbol{c}_l$). Thus, if all $\boldsymbol{c}
_l$ with $l<k$ are uniform, the only task left is to set $k$UB$_x[J]$ uniform
for all $x$. Similarly, as in the proof of Lemma
\ref{lem:ck_constant}, we see that this uniform value is
\begin{equation}\label{eq:ck_newterm_uniform}
k\text{UB}_x[J]=\frac{(k-1)!}{d^{k-1}}
~~~\forall x_{y_1},\ldots,x_{y_k}\in[d],~y_1,\ldots,y_k\in J
\end{equation}
for all $J\subseteq\{1,\ldots,n\}$, $\vert J\vert=k$. This is because
\begin{enumerate*}[label=\itshape\alph*\upshape)] \item
there are $(k-1)!$ $k$-cycles in $S_k$, \item as we saw in Lemma
\ref{lem:ck_constant}, $\sum_x\prod_{y\in J}
\langle y_{x_y}\lvert{\sigma(y)}_{x_{\sigma(y)}}\rangle=d$ if $\sigma$ is a 
$k$-cycle, and \item there
are $d^k$ different $k$UB$_x[J]$s for all possible $x_{y_1},\ldots,x_{y_k}
\in[d]$ \end{enumerate*}. At this point, I note
that it should not be necessary to set the same values for all the possible
$J$ subsets, but eventually it will turn out that this should be the case.
I also refer the reader to Definition \ref{def:nUB}, from where it's clear
that the condition in Eq. (\ref{eq:ck_newterm_uniform}) is equivalent to
saying that in an optimal strategy, every subset $J$ of size $k$ of our
measurement bases should form
a $k$UB. The following theorem on $n$UBs simplifies the criteria on optimal
QRAC measurements, and finalize the proof of Theorem \ref{thm:main}:
\begin{thm}\label{thm:n_to_n-1}
If $n$ orthonormal bases form an $n$UB, then any subset of $n-1$ bases forms an
$(n-1)$UB.
\end{thm}
\begin{proof}
Remember that the $n$ bases $\{\ket{y_{x_y}}\}_{x_y=1}^d$ for $y=1,\ldots,n$
form an $n$UB iff
\begin{equation}\label{eq:nUB2}
\sum_{\substack{\sigma\in S_n \\ \sigma\text{: n-cycle}}}\prod_{y=1}^n
\braket{y_{x_y}}{\sigma(y)_
{x_{\sigma(y)}}}=\frac{(n-1)!}{d^{n-1}} ~~~ \forall x_1,\ldots,x_n\in[d].
\end{equation}
Pick an arbitrary subset of $n-1$ bases by omitting one basis, say $\{\ket{k_
{x_k}}\}_{x_k=1}^d$, for some $k\in\{1,\ldots,n\}$. Now, sum up Eq.
(\ref{eq:nUB2}) over the index $x_k$. This, as we saw in the proof of Lemma
\ref{lem:ck_constant}, eliminates the term $\ket{k_{x_k}}\bra{k_{x_k}}$ from
each term of the summand, as $\sum_{x_k}\ket{k_{x_k}}\bra{k_{x_k}}=\unit$. The
summation on the RHS of Eq. (\ref{eq:nUB2}) clearly yields $\frac{(n-1)!}
{d^{n-2}}$, thus the equation now reads as:
\begin{equation}\label{eq:nUBsum}
\sum_{\substack{\sigma\in S_n \\ \sigma\text{: n-cycle}}}\prod_{\substack{y=1 \\
y\ne k}}^n
\braket{y_{x_y}}{\sigma(y)_
{x_{\sigma(y)}}}=\frac{(n-1)!}{d^{n-2}} ~~~ \forall\{x_1,\ldots,x_n\}\setminus\{x_k\}
\in[d].
\end{equation}
We can write the LHS of the above equation (reordering the product to a
desirable form) as:
\begin{equation}\label{eq:nUBsum2}
\begin{split}
\sum_{\substack{\sigma\in S_n \\ \sigma\text{: n-cycle}}} & \left.
\big\langle\sigma(k)_{x_{\sigma(k)}}
\big\lvert \sigma^2(k)_{x_{\sigma^2(k)}}\big\rangle\big\langle\sigma^2(k)_
{x_{\sigma^2(k)}}\big\lvert\cdots \right. \\
& \left. \cdots \big\lvert\sigma^{n-1}(k)_{x_{\sigma^{n-1}(k)}}\big\rangle
\big\langle\sigma^{n-1}(k)_{x_{\sigma^{n-1}
(k)}}\big\lvert\sigma(k)_ {x_{\sigma(k)}}\big\rangle. \right.
\end{split}
\end{equation}
The permutation $\sigma$ can be represented by the chain of elements, up to a
cyclic permutation as
\begin{equation}\label{eq:sigmabig}
\sigma=[\sigma(k),\sigma^2(k),\ldots,\sigma^{n-1}(k),k].
\end{equation}
Note that in Eq. (\ref{eq:nUBsum2}), we only use the first $n-1$ elements of
this representation, since the index $k$ is not present anymore. In fact, for
our purposes, we can look at the first $n-1$ elements as an artificial
$(n-1)$-cycle $\sigma_\rvert$, and write Eq. (\ref{eq:nUBsum2}) as
\begin{equation}\label{eq:nUBsum3}
\sum_{\substack{\sigma\in S_n \\ \sigma\text{: n-cycle}}}\prod_{\substack{y=1 \\
y\ne k}}^n
\braket{y_{x_y}}{\sigma_\rvert(y)_
{x_{\sigma_\rvert(y)}}}.
\end{equation}
Nevertheless,
the summation still runs along all the $(n-1)!$ $n$-cycles, and not the new
$(n-1)$-cycles. What is the connection between these? Let's consider 
new $n$-cycles $\sigma_\pi$, by cyclically permuting the first $n-1$
elements in the representation Eq. (\ref{eq:sigmabig}) of $\sigma$:
\begin{equation}\label{eq:sigmapi}
\sigma_\pi=\big[\pi\big(\sigma(k)\big),\pi\big(\sigma^2(k)\big),\ldots,
\pi\big(\sigma^{n-1}(k)\big),k\big],
\end{equation}
where $\pi$ is a cyclic permutation of the elements $\sigma(k),\sigma^2(k),
\ldots,\sigma^{n-1}(k)$. There are $n-1$ different such
$\sigma_\pi$ permutations, each of which giving a new $n$-cycle. Nevertheless,
$\sigma_{\pi\rvert}$, the artificial $(n-1)$-cycles given by their first $n-1$
elements are all the same $\sigma_{\pi\rvert}\equiv\sigma_\rvert$, since they
are only cyclic permutations of each other. Thus, when summing up over all
$n$-cycles in Eq. (\ref{eq:nUBsum3}), we use each $(n-1)$-cycle $n-1$ times,
and hence we can write it as 
\begin{equation}\label{eq:nUBsum4}
(n-1)\sum_{\substack{\sigma\in S_{n-1} \\ \sigma\text{: (n-1)-cycle}}}\prod_
{\substack{y=1 \\ y\ne k}}^n \braket{y_{x_y}}{\sigma(y)_{x_{\sigma(y)}}},
\end{equation}
and comparing this with Eq. (\ref{eq:nUBsum}) we see that
\begin{equation}\label{eq:nUBsumfinal}
\sum_{\substack{\sigma\in S_(n-1) \\ \sigma\text{: (n-1)-cycle}}}\prod_
{\substack{y=1 \\ y\ne k}}^n \braket{y_{x_y}}{\sigma(y)_
{x_{\sigma(y)}}}=\frac{(n-2)!}{d^{n-2}} ~~~ \forall\{x_1,\ldots,x_n\}\setminus
\{x_k\} \in[d],
\end{equation}
i.e. the set $\{\ket{y_{x_y}}\}_{x_y=1}^d$ for $y\in\{1,\ldots,n\}\setminus
\{k\}$ forms an $(n-1)$UB.
\end{proof}
\begin{remark}
Note that Theorem \ref{thm:n_to_n-1} implies that for an optimal $n^d\to1$ QRAC
strategy it is enough to have $n$UB measurement bases, since the $n$UB condition
implies the uniformity of $\boldsymbol{c}_n$ and in fact the uniformity of all
$\boldsymbol{c}_k$. Thus, the proof of Theorem \ref{thm:main} is complete.
\end{remark}

\section{Existence of $n$-fold unbiased bases}\label{sec:existence}
So far, I only established a theoretical optimum for QRAC strategies. But
naturally the question arises: do these optimal measurement bases exist for any
dimension? In the $2^d\to1$ case, we are always provided with a pair of MUBs in
any dimension. It will turn out that this is not the case with $n$UBs in
general, although their existence problem is still unsolved.

\subsection{Low dimensions}

Theorem \ref{thm:n_to_n-1} provides a useful tool for searching for $n$UBs.
It implies that any subset of size $n-1$ of a set of $n$UBs should also form
an $(n-1)$UB, and thus eventually they should all form MUBs. MUBs are
excessively studied in the quantum information community (see \cite{onMUBs,
MUB2-5}), hence making the search for $n$UBs more tractable.

The easiest non-trivial search is for 3UBs. Following the above argument, every
3UB should be a triplet of MUBs. All the triplets are fully characterised in
dimensions 2, 3, 4 and 5, allowing for an exhaustive search \cite{MUB2-5}.
For the readers' convenience, let me recall a set of equivalence transformations
on MUBs \cite[Appendix A]{MUB2-5}, which are easily seen equivalence
transformations also on $n$UBs for any $n$:
\begin{defn}\label{defn:MUB_eqv}
Consider a set of $r$ MUBs ($n$UBs) described by complex matrices $B_i$,
$i=1,\ldots,r$ of size $d\times d$, that is, the elements of basis $i$ are the
columns of $B_i$. Two such lists are \textit{equivalent} to
each other, if they can be transformed into each othey by a succession of the
following five transformations:
\begin{enumerate}
\item an \textit{overall unitary} transofmation $U$ applied from the left,
\begin{equation}\label{eq:eqv_global_U}
\{B_1,\ldots,B_r\}\to U\{B_1,\ldots,B_r\},
\end{equation}
which leaves invariant all the scalar products.
\item $r$ \textit{diagonal unitary} transformations $D_i$ from the right which
attach phase factors to each column of the $r$ matrices,
\begin{equation}\label{eq:eqv_global_U}
\{B_1,\ldots,B_r\}\to \{B_1D_1,\ldots,B_rD_r\}.
\end{equation}
This exploits the fact that the overall phase of a quantum state drops out from
the conditions of MUBs ($n$UBs).
\item $r$ \textit{permutations} of the elements within each basis,
\begin{equation}\label{eq:eqv_global_U}
\{B_1,\ldots,B_r\}\to \{B_1P_1,\ldots,B_rP_r\},
\end{equation}
which is simply just relabeling basis elements. Here $P_i$ are unitary
permutation matrices, $P_iP_i^T=\unit$.
\item \textit{pairwise exchange} of two bases,
\begin{equation}\label{eq:eqv_global_U}
\{\ldots,B_i,\ldots,B_j,\ldots\}\to \{\ldots,B_j\ldots,B_i,\ldots\}.
\end{equation}
whish is simply relabeling the bases.
\item an \textit{overall complex conjugation}
\begin{equation}\label{eq:eqv_global_U}
\{B_1,\ldots,B_r\}\to \{\bar{B}_1,\ldots,\bar{B}_r\},
\end{equation}
which leaves invariant all the scalar products.
\end{enumerate}
\end{defn}

All equivalence classes or MUB triplets are known in dimensions 2, 3, 4 and 5.
After checking these triplets for the 3UB condition, I got the following
results:

In dimension 2, there is only one equivalence class of MUB triplets, which also
forms a 3UB. A representative of this class is:
\begin{equation}
\Bigg\{
\begin{pmatrix}
1 & 0 \\
0 & 1
\end{pmatrix},
\frac{1}{\sqrt{2}}
\begin{pmatrix}
1 & 1 \\
1 & -1
\end{pmatrix},
\frac{1}{\sqrt{2}}
\begin{pmatrix}
1 & 1 \\
i & -i
\end{pmatrix}
\Bigg\}.
\end{equation}
Although this is promising, observe that the existence of 3UBs in dimension 2 is
trivial in some sense. We already know that there exist 3 MUBs. Then, using
them, in the
characteristic polynomial (\ref{eq:charpoly_general}), the uniformity of $c_3$
gives the 3UB condition. But this is the determinant of the Gramian matrix
of 3 vectors for each $x$. These vectors must be linearly dependent in dimension
2, and we know that in this case, the Gram determinant is zero (see e.g.
\cite[Theorem 7.2.10]{gram}). Thus, the coefficients $(c_3)_x$ are uniform for
every $x$
(zero, in fact). Also note, that this set of 3UBs is not useful in the QRAC
game, as my argument only works for $d\ge n$. Nevertheless, these measurement
bases give the optimal strategy for the $3^2\to1$ QRAC \cite{QRACSR}.
The above observations can be generalized to the following result on $n$UBs:
\begin{prop}\label{prop:d+1UB}
If there exist $d+1$ $d$UBs in dimension $d$, then they also form a $(d+1)$UB.
\end{prop}
\begin{proof}
Assuming that the $d+1$ bases are $d$UBs, in the coefficients $(c_{d+1})_x$, 
all the terms are uniform, except for the $(d+1)$UB term. But $(c_{d+1})_x$ is 
a Gram determinant of $d+1$ vectors in dimension $d$, thus equals zero for
every $x$.
\end{proof}
\begin{remark}
It is known that in dimension $d$, the maximal number of MUBs is $d+1$. Also
note that if $d>2$, the $d$UB conditions on a set of $d+1$ bases are more
restrictive than the MUB conditions. It is then unlikely that in any
dimension greater than 2, there exist $d+1$ $d$UBs, and thus Proposition
\ref{prop:d+1UB} may practically only be useful in the $d=2$ case.
\end{remark}
Since there are no more than 3 MUBs in dimension 2, there cannot exist $n$UBs
with $n\ge4$.

In dimension 3, there is also only one equivalence class of MUB triplets. If
we check for the 3UB condition, it does not satisfy it, meaning that there are
no $n$UBs in dimension 3 for $n\ge3$.

In dimension 4, there is a three-parameter family of MUB triplets. Nevertheless,
it turns out that they do not satisfy the 3UB condition for any value of these
parameters, thus concluding that there are no $n$UBs in dimension 4 for $n\ge3$.

In dimension 5, there are two equivalence classes of MUB triplets, none of them
satisfying the 3UB condition, meaning that also in dimension 5, there are no
$n$UBs for $n\ge3$.

\subsection{High dimensions, probabilistic arguments}

In dimensions higher than 5, not all the equivalence classes of MUB triplets
are known. There are some explicit constructions in prime and prime power
dimensions for obtaining a full set of $d+1$ MUBs, from which we can test
arbitrary subsets for the $n$UB conditions. In composite dimensions that are not
powers of primes, even the maximal number of MUBs is unknown. The lowest such
dimension, 6 is excessively studied, pointing to a direction that there exist
only 3 MUBs \cite{MUB6_3,MUB6_4,MUB6_5,MUB6_6}.

In dimension six, I checked a one-parameter family of MUB triplets
\cite{MUB6_1}, which lead to no success. Apart from that, I checked
some known constructions for dimensions 7, 8 and 9 \cite{onMUBs,MUB789}, also
without any success. Nevertheless, these searches weren't exhaustive, given the
fact that there is no characterization of all the equivalence classes of MUB
triplets in these dimensions.

The situation thus seems a little desperate at this point. On the bright side,
when we move to high dimensions, there are some probabilistic arguments
supporting the possibility of existence of $n$UBs. Consider $n$ uniformly
random states on $\mathbb{C}^d$, constructed as follows: fix a state $\ket{x_1}$
in the computational basis. Then draw unitaries $U_2,U_3,\ldots,U_n$ uniformly
and independently
with respect to the Haar measure on $\mathbb{U}_d$, and apply it to the other
states of the computational basis: $U_y\ket{x_y}$ with $y=2,3,\ldots,n$. We will
be interested in the expectation value of the $n$UB expression Eq.
(\ref{eq:nUB}) for these states, with the correspondence $\ket{y_{x_y}}:=
U_y\ket{x_y}$:
\begin{equation}\label{eq:nUBrnd}
\mathbb{E}\Big(
\sum_{\substack{\sigma\in S_n \\ \sigma\text{: n-cycle}}}\prod_{y=1}^n
\langle{x_y}\rvert U_y^\dagger U_{\sigma(y)}\lvert x_{\sigma(y)}\rangle\Big),
\end{equation}
where the expectation value is over $\mathbb{U}_d$ with the Haar measure, and we
say that $U_1=\unit$.

Since the expectation value is linear, we can concentrate on one term in the
above sum. Pick the term with $\sigma=[1,2,\ldots,n]$, and write out the
expectation value:
\begin{equation}\label{eq:nUBexp}
\begin{split}
& \left. \mathbb{E}\Big(
\langle{x_1}\rvert U_2\lvert x_2\rangle\langle x_2\rvert U_2^\dagger U_3\lvert
x_3\rangle\langle x_3\rvert U_3^\dagger \cdots  U_n\lvert
x_n\rangle\langle x_n\rvert U_n^\dagger\lvert x_1\rangle\Big) \right. \\
= & \left. \mathbb{E}\bigg(
\langle{x_1}\rvert U_2\lvert x_2\rangle\langle x_2\rvert U_2^\dagger 
\sum_{k_2=1}^d\lvert k_2\rangle\langle k_2\rvert U_3\lvert
x_3\rangle\langle x_3\rvert U_3^\dagger 
\sum_{k_3=1}^d\lvert k_3\rangle\langle k_3\rvert \cdots \right. \\
& \left. \sum_{k_{n-1}=1}^d\lvert k_{n-1}\rangle\langle
k_{n-1}\rvert U_n\lvert x_n\rangle \langle x_n\rvert U_n^\dagger\lvert x_1
\rangle\bigg) \right. \\
= & \left. \sum_{k_2,k_3,\ldots,k_{n-1}=1}^d\mathbb{E}\big(\langle{x_1}\rvert
U_2 \lvert x_2\rangle\langle x_2\rvert U_2^\dagger\lvert k_2\rangle\big)
\mathbb{E}\big(\langle{k_2}\rvert U_3 \lvert x_3\rangle\langle x_3\rvert U_3^
\dagger\lvert k_3\rangle\big)\cdots \right. \\
& \left. \cdots\mathbb{E}\big(\langle{k_{n-1}}\rvert U_n \lvert x_n\rangle
\langle x_n\rvert U_n^ \dagger\lvert x_1\rangle\big), \right.
\end{split}
\end{equation}
inserting identities in between the different
unitaries in the first equality, and using the linearity of the expectation
value and the fact that the unitaries are independently drawn, in the second
one. Now, let's calculate in general the above expectation values:
\begin{equation}\label{eq:termexp}
\begin{split}
& \left. \mathbb{E}\big(\langle x_i\rvert U_y \lvert x_y\rangle\langle x_y
\rvert U_y^\dagger\lvert x_j\rangle\big) = \int \langle x_i\rvert U \lvert x_y
\rangle\langle x_y \rvert U^\dagger\lvert x_j\rangle \text{d}U =\langle x_i
\rvert \Big(\int U \lvert x_y \rangle\langle x_y \rvert U^\dagger
\text{d}U\Big)\lvert x_j\rangle \right. \\
= & \left. \langle x_i\rvert\frac{\unit}{d}\lvert x_j\rangle = \frac1d \langle 
x_i\lvert x_j\rangle=\frac{\delta_{ij}}{d}, \right.
\end{split}
\end{equation}
where d$U$ is the Haar measure, and the linearity of the inner product is used.
Plugging this into Eq. (\ref{eq:nUBexp}), we get that
\begin{equation}\label{eq:nUBexpres}
\begin{split}
& \left. \mathbb{E}\Big(
\langle{x_1}\rvert U_2\lvert x_2\rangle\langle x_2\rvert U_2^\dagger U_3\lvert
x_3\rangle\langle x_3\rvert U_3^\dagger U_4\cdots U_{n-1}^\dagger U_n\lvert
x_n\rangle\langle x_n\rvert U_n^\dagger\lvert x_1\rangle\Big) \right. \\
= & \left. \sum_{k_2,k_3,\ldots,k_{n-1}=1}^d\frac{1}{d^{n-1}}
\langle x_1\lvert k_2\rangle\langle k_2\rvert k_3\rangle\langle k_3\rvert
\cdots\lvert k_{n-1}\rangle\langle k_{n-1}\rvert x_1\rangle = \frac{1}{d^{n-1}},
\right.
\end{split}
\end{equation}
and observe that for any $\sigma$ $n$-cycle, we get the same value. Thus, we can
conclude, that
\begin{equation}\label{eq:nUBexpfinal}
\mathbb{E}\Big(
\sum_{\substack{\sigma\in S_n \\ \sigma\text{: n-cycle}}}\prod_{y=1}^n
\langle{x_y}\rvert U_y^\dagger U_{\sigma(y)}\lvert x_{\sigma(y)}\rangle\Big)
=\frac{(n-1)!}{d^{n-1}}.
\end{equation}
This means that the expectation value of the $n$UB expression for $n$ 
independent, uniformly random states is exactly the $n$UB condition. This is
promising, as it implies that a set of states forming $n$UBs could even be
something typical.

In general, one would then prove existence by using the probabilistic method
\cite{probmethod}, thus invoking concentration of measure, in particular,
L\'evy's lemma (see e.g. \cite[Theorem 4]{levy1} or \cite{levy2}):
\begin{lem}[L\'evy's lemma]\label{lem:levy}
Let $f:S^{2d-1}\to\mathbb{R}$ be Lipschitz-continuous with Lipschitz constant
$\eta$, i.e.
\begin{equation}\label{eq:Lipschitz}
\lvert f(x)-f(y)\rvert \le\eta\cdot \lVert x-y\rVert,
\end{equation}
where $\lVert . \rVert$ is the Eucledian norm on $\mathbb{R}^{2d}$. Then, drawing
a point $x\in S^{2d-1}$ randomly with respect to the uniform measure on the
sphere yields
\begin{equation}\label{eq:Levy}
\mathbb{P}(\lvert f(x)-\mathbb{E}_f\rvert\ge\epsilon)\le2\text{exp}\bigg
(\frac{d\epsilon^2}{9\pi^3\eta^2}\bigg)
\end{equation}
for all $\epsilon\ge0$.
\end{lem}
This means that for a real-valued function on pure states of dimension $d$,
the probability of deviating from its expectation value decreases exponentially
with increasing dimension, provided an appropriate Lipschitz constant.

The function on $d$-dimensional pure states, whose expectation value are
calculated in Eq. \eqref{eq:termexp} is $f_{ij}:\ket{\psi}\to\braket{x_i}{\psi}
\braket{\psi}{x_j}$. Now, even though its expectation value is real, in general
it is not a real-valued function, thus the above lemma does not apply.
Also, if one wants to see uniformity of the
expression $n$UB$_x$ for all $x$, the states are not independent anymore, as
certain subsets have to form orthogonal bases.

To sum it up, the fact that the expectation value of the $n$UB expression is
what we want it to be is promising. Although, to prove existence, the usual
probabilistic method faces difficulties. Some more refined concentration of
measure results on complex valued functions, or functions on unitaries would be
needed, if one wanted to prove existence this way.

\section{Applications}\label{sec:applications}

In this section, I give a few (potential) applications of the above defined
bases. We certainly know a lot about their implications on QRAC strategies now,
and I outline some possible applications on other tasks otherwise related to
MUBs. Apart from the protocols mentioned here, one could consider other tasks
usually discussed in the context of MUBs. Note that if the bases in question
don't exist, they still provide a
bound on what can we achieve within the framework of quantum mechanics, and
this bound is close-to-tight, at least in the QRAC scenario.

\subsection{Upper bounds on QRAC success probabilities}

Naturally, $n$UBs in dimension $d$ provide optimal $n^d\to1$ QRAC measurements,
as long as
$d\ge n$. In the case when they exist, this is a tight upper bound on quantum
strategies. Nevertheless, in the case when they don't exist, they still give a
close-to-tight upper bound. To demonstrate this, I provide a table with the
ASP of the optimal classical, the MUB and the $n$UB quantum
strategies for some simple cases:

\begin{center}
\begin{tabular}{| c | c | c | c | c |}
\hline
\multicolumn{2}{|c|}{ } & classical & MUB & $n$UB \\ \hline
$d=3$ & $n=3$ & 0.6296 & 0.6971 & 0.6989  \\ \hline
\multirow{2}{*}{$d=4$} & $n=3$ & 0.5625 & 0.6443 & 0.6466 \\ \cline{2-5}
& $n=4$ & 0.5313 & 0.5779 & 0.5872 \\ \hline
\multirow{2}{*}{$d=5$} & $n=3$ & 0.5200 & 0.6109 & 0.6114 \\ \cline{2-5}
& $n=4$ & 0.4880 & 0.5430 & 0.5477 \\ \hline
%& $n=5$ & 0.4573 & 0 & 0 \\ \hline
\end{tabular}
\end{center}

Here, the classical values are computed using the method of \cite{binaryRAC},
the MUB values are results of straightforward calculations exploiting the
known equivalence classes, and the $n$UB values are computed by calculating
$\lambda_x^{\text{max}}$, assuming uniform coefficients in $x$ for every $k$ in
the characteristic polynomial, Eqs. \eqref{eq:charpoly_general}, \eqref{eq:ck}.

On the other hand, I note that the see-saw optimization, also used in
\cite{QRACMUB}, results in MUB measurements for $n=3$, in dimensions 3-7.
This serves as numerical evidence for MUB optimality whenever $n$UBs don't
exist, and also for the non-existence of 3UBs in dimensions 6, 7. An open
question remains whether MUBs provide optimal measurements for any $n^d\to1$
QRAC protocol.

Remember that any pair of $n$UBs also form MUBs, thus in this sense, MUBs do
provide optimal measurements for general QRACs. Although, when $n$UBs don't
exist in the given dimension, the question of optimal QRAC strategies becomes
more difficult. To see this, consider the optimization of a
$3^d\to1$ QRAC in some dimension where 3UBs don't exist. This means that the
vector $\boldsymbol{c}_3$ cannot be set uniform. On the other hand, there exist
3 MUBs in any dimension, thus $\boldsymbol{c}_2$ can always be set uniform.
Although, this doesn't imply that MUBs are optimal in this case, as the
ASP $\bar{p}(\boldsymbol{c}_2,\boldsymbol{c}_3)$ is a
Schur-concave function of $\boldsymbol{c}_2$ and $\boldsymbol{c}_3$, but this
is an independent property on the two vectors. This then only
means that $\bar{p}':=\bar{p}(\boldsymbol{c}'_2,\boldsymbol{c}'_3)<\bar{p}$,
whenever $\boldsymbol{c}_2\succ\boldsymbol{c}'_2$ \textit{and} $\boldsymbol{c}_3
\succ\boldsymbol{c}'_3$, where $\succ$ expresses majorization (see e.g.
\cite{inequalities}). We cannot say anything about the relation of $\bar{p}'$
and $\bar{p}$, when e.g. $\boldsymbol{c}_2\succ\boldsymbol{c}'_2$ and
$\boldsymbol{c}_3\prec\boldsymbol{c}'_3$, and one can construct bases such that
these relations hold. Nevertheless, the above mentioned numerical evidence
supports the optimality of MUBs for $3^d\to1$ QRACs.

I note, as it's pointed out in \cite{QRACMUB}, that if we restrict
ourselves to MUB optimization, then different equivalence classes can yield
different ASPs. The simplest case is $d=5$, $n=3$,
where the two inequivalent MUB triplets perform differently. This is because
there don't exist 3UBs in dimension 5, but the two equivalence classes have
different 3UB properties, and the one being more uniform gives a better QRAC
strategy.

Since polynomials up to order 4 are analytically solvable, the $n$UB method
provides analytic bounds for $n^d\to1$ QRAC ASPs for $n=2,3,4$
with $d\ge n$. I note that for $n=2$, this bound is tight and is previously
found in \cite{QRACMUB}:
\begin{equation}\label{eq:2to1bound}
\bar{p}_{n=2}\le\frac12\Big(1+\frac{1}{\sqrt{d}}\Big)
\end{equation}
\begin{equation}\label{eq:3to1bound}
\bar{p}_{n=3}\le\frac13\Bigg(1+\frac{d}{\Big(d^4+\sqrt{d^8-d^9}\Big)^{1/3}}
+\frac{\Big(d^4+\sqrt{d^8-d^9}\Big)^{1/3}}{d^2}\Bigg).
\end{equation}
The formula for $n=4$ is too complicated to present here, but it is
the greatest zero of the polynomial
\begin{equation}\label{eq:4to1eq}
\lambda^4-4\lambda^3+6\bigg(1-\frac1d\bigg)\lambda^2-4\bigg(1-\frac3d+\frac{2}
{d^2}\bigg)\lambda+1- \frac6d+\frac{11}{d^2}-\frac{6}{d^3},
\end{equation}
divided by 4.

If $n>4$, the polynomials are not solvable analytically anymore, nevertheless,
one can solve them numerically up to machine precision (e.g. by Newton's method,
with starting point $n$). This gives an upper bound on the given $n^d\to1$
QRAC ASP. Also, if $n$ MUBs exist in dimension $d$, they give a lower bound.
In low dimensions, these bounds are close-to-tight, and this is expected in
higher dimensions as well, thus one has a good estimate on optimal QRAC ASPs for
a wide class of $n$ and $d$.

\subsection{Entropic uncertainty relations}

Entropic uncertainty relations are a refined version of Heisenberg's uncertainty
relations (see the seminal paper of Maassen and Uffink \cite{maassen_uffink},
and the surveys \cite{uncertainty_survey,uncertainty_survey2}).
Consider $n$ observables on $\mathbb{C}^d$, described by projections on the
states $\{\ket {y_{x_y}}\}_{x_y=1}^d$, $y\in[n]$. Then define the
probabilities $p^y_{x_y}=\lvert
\braket{y_{x_y}}{\psi}\rvert^2$ for some $\ket{\psi}\in\mathbb{C}^d$. The aim is
then to put a lower bound on $\sum_{y=1}^nH(\{p^y_{x_y}\})$, where $H$ is the
Shannon entropy. It is shown in \cite{maassen_uffink}, that when $n=2$,
\begin{equation}\label{eq:Maassen_Uffink}
H(\{p^y_{x_y}\})+H(\{p^z_{x_z}\})\ge-\log{c},
\end{equation}
where $c=\max_{x_y,x_z}\lvert\braket{y_{x_y}}{z_{x_z}}\rvert^2$. This bound is
independent on the state $\ket{\psi}$, and the lowest possible value of $c$ is
attained by mutually unbiased bases, for which $c=\frac1d$.

In the case of more than two observables, no general tight bound is known.
Based on the fact that the bound for $n=2$ is related to the uniformity of
MUBs, I propose that $n$UBs could provide potential means of exploring and
understanding entropic uncertainty relations for $n$ observables.

\subsection{Information locking}

An information theoretical task closely related to entropic uncertainty
relations is that of information locking (see \cite{locking1,locking2} for
detailed
description). Here, classical correlations are hidden (locked) in quantum
states, until a key is revealed. It turns out that by revealing this extra
information, arbitrarily large increase can be obtained in the correlations. In
the simplest case, this means that one party is encoding a classical dit in a
qudit, using one of two mutually unbiased bases. Sending this qudit, but not the
information on the encoding basis (one bit key) to a receiver leaves them with
very limited classical correlation, since measuring in the wrong basis provides
no information on the encoded dit whatsoever. Sending the key, on the other
hand reveals the full information, thus increases classical correlation to its
maximal value.

It is known that for a one-bit key, corresponding to two possible encoding
bases, mutually unbiased bases provide optimal locking properties. I
propose that in the case of $n$ possible encoding bases, $n$UBs could provide
close-to-tight bounds on locking tasks.

\section{Foundational implications}\label{sec:foundations}

Apart from their use in information theoretical protocols, and providing bounds
on certain tasks, the question of existence of $n$UBs raises some fundamental
questions about the quantum world. One of these is solely the structure of
quantum states, which we still strive to understand, especially in higher
dimensions. On the other hand, considering the fact that in the QRAC scenario,
$n$UBs are extremely close to what is achievable within quantum mechanics, it
is natural to ask if their existence is prohibited merely by the formulation
of quantum mechanics, or is it some fundamental property of Nature. A
foundational question seemingly well-fit for investigating this problem is that
of the existence of genuine high-order interference in Nature.

\subsection{Geometry of quantum states}

While the mathematical formulation of quantum states is clear, we are struggling
to characterize their geometry, especially in dimensions higher than two, where
the Bloch-sphere ceases to provide an intuitive picture. Whenever we impose some
conditions on certain states, such as the MUB conditions, we have an option to
characterize quantum states accordingly. For instance, the long-standing
question of the number MUBs in a general dimension 
could allow us to characterize the behaviour of quantum states in
different dimensions, according to the unbiasedness one can introduce in
certain protocols. In the same spirit, understanding how $n$UBs can or cannot be
constructed, could give a more general characterisation.

\subsection{Genuine $n$-th order interference}

It was noted by Sorkin \cite{sorkin}, that quantum mechanics only
exhibits second-order genuine interference. Simply saying, having a two-slit
experiment with quantum particles, the interference pattern cannot be written
in terms of one-slit experiments. On the other hand, already a three-slit
experiment can be written in terms of one- and two-slit experiments. This
follows simply from the mathematical formulation of quantum mechanics. The
natural question then arises, whether Nature admits genuine higher-order
interference, or if not, what is the fundamental reason behind it. It is worth
to note that even if there is higher-order interference, there is
experimental evidence that it is suppressed by at least a factor of $\sim10^2$
by second-order interference \cite{interference_experiment1,
interference_experiment2}.

For studying this question, researchers have come up with theories (in general,
general probabilistic theories) that exhibit genuine higher-order interference
(see e.g. \cite{density_cubes,qqt}, or \cite{higher-order} for a review on
them). For now, I will focus on the theory of Density Cubes of
Daki\'c et al. \cite{density_cubes}. They point out that the description of
quantum states by density matrices $\rho_{ij}$ inherently only allows for
interference between two levels of a quantum state. To overcome this limitation,
they introduce Density Cubes, that is, states described by 3-index tensors,
$\rho_{ijk}$. They construct some (incomplete) bases, and show that usual
quantum states form a subset of these generalized state space. 
Nevertheless, it is pointed out in \cite{higher-order} that the axioms of this
theory are insufficient to uniquely characterise it.

In any case, if one considers a set of 3UBs, $\{\ket{y_{x_y}}\}_{x_y}$,
$\{\ket{z_{x_z}}\}_{x_z}$, $\{\ket{a_{x_a}}\}_{x_a}$, then in the 3UB condition
\begin{equation}\label{eq:3UB}
\braket{y_{x_y}}{z_{x_z}}\braket{z_{x_z}}{a_{x_a}}\braket{a_{x_a}}{y_{x_y}}+
\braket{y_{x_y}}{a_{x_a}}\braket{a_{x_a}}{z_{x_z}}\braket{z_{x_z}}{y_{x_y}}=
\frac{2}{d^2}~~~\forall x_y,x_z,x_a\in[d]
\end{equation}
correlations of 3 ``levels'' appear. It is then natural to think that if 3UBs
don't exist in some dimension within the framework of quantum mechanics, some
analogue might exist within the Density Cube framework. Note that the idea of
Density Cubes and this analogue can be generalized to any $n$ other than 3.

Existence of $n$UBs thus could be connected to the existence of $n$-th order
interference in Nature. Then, understanding 
the fundamental reasons why these theories could or could not describe Nature,
could lead to understanding the existence of $n$UBs. Or the other way around,
understanding the existence problem of $n$UBs could lead to non-trivial
statements on $n$-th order interference in Nature.

\section*{Acknowledgements}
I would like to thank Jakub Borka\l a for providing calculations on
general QRAC strategies, Edgar Aguilar and Richard K\"ung for fruitful
discussions, and Piotr Mironowicz and Debashis Saha for providing numerical
results on some QRAC protocols. The work is supported by the NCN grant
Sonata UMO-2014/14/E/ST2/00020

\bibliography{nUBbib}

\bibliographystyle{ieeetr}

\end{document}